\documentclass[12pt]{article} 
\usepackage{graphicx} 
\usepackage{epsf} 
\usepackage{epsfig} 
\usepackage{setspace} 
\usepackage{fullpage} 
\usepackage[round]{natbib}  
\usepackage{lscape} 
\usepackage{rotating} 
\usepackage{longtable} 
\usepackage{float} 
\setcitestyle{aysep={}}

 \usepackage{times}

\RequirePackage{geometry}
\usepackage{multirow}
\usepackage{booktabs}
\usepackage{dcolumn}
\usepackage{ragged2e}

\usepackage{amsmath}
\usepackage{amsfonts}
\usepackage{amsthm}
\usepackage{bbm}
\usepackage{array}
\newtheorem{lemma}{Lemma}
\newtheorem{prop}{Proposition}
\newtheorem{assumption}{Assumption}

\usepackage{dsfont}

\def\Var{{\rm Var}\,}
\def\E{{\rm E}\,}

\def\N{{\rm N}\,}

\newcommand{\indicator}[1]{\mathds{1}\left\{ #1 \right\}}

\newcommand{\plim}{\text{plim} \hspace{0.05in}}

\newcommand\independent{\protect\mathpalette{\protect\independenT}{\perp}}
\def\independenT#1#2{\mathrel{\rlap{$#1#2$}\mkern2mu{#1#2}}}

\begin{document}

\sloppy

    \renewcommand{\topfraction}{0.9}	
    \renewcommand{\bottomfraction}{0.8}	
    \setcounter{topnumber}{2}
    \setcounter{bottomnumber}{2}
    \setcounter{totalnumber}{4}     
    \setcounter{dbltopnumber}{2}    
    \renewcommand{\dbltopfraction}{0.9}	
    \renewcommand{\textfraction}{0.07}	
    \renewcommand{\floatpagefraction}{0.7}	
    \renewcommand{\dblfloatpagefraction}{0.7}	

\def\ci{\perp\!\!\!\perp}

\title{Combining List Experiment and Direct Question Estimates of Sensitive Behavior Prevalence} 

\author{Peter M. Aronow, Alexander Coppock, Forrest W. Crawford, \\ and Donald P. Green\footnote{Peter M. Aronow is Assistant Professor, Department of Political Science, Yale University, 77 Prospect Street, New Haven, CT 06520. 
Alexander Coppock is Doctoral Student, Department of Political Science, Columbia University, 420 W. 118th Street, New York, NY 10027. 
Forrest W. Crawford is Assistant Professor, Department of Biostatistics, Yale School of Public Health, 60 College Street, New Haven, CT 06520.
Donald P. Green is Professor of Political Science at Columbia University, 420 W. 118th Street, New York, NY 10027. 
The authors are grateful to Columbia University, which funded components of this research but bears no responsibility for the content of this report. This research was reviewed and approved by the Institutional Review Board of Columbia University (IRB-AAAL2659). Helpful comments from Xiaoxuan Cai, Albert Fang, Cyrus Samii, Chris Schuck, and Michael Schwam-Baird are greatly appreciated.}} 
\maketitle \thispagestyle{empty}

\RaggedRight
\setlength{\parindent}{15pt}
\bibliographystyle{apsr}

\begin{abstract}

\normalsize{Survey respondents may give untruthful answers to sensitive questions when asked directly. In recent years, researchers have turned to the list experiment (also known as the item count technique) to overcome this difficulty. 
While list experiments may be less prone to bias than direct questioning, list experiments are also more susceptible to sampling variability. We show that researchers do not have to abandon direct questioning altogether in order to gain the advantages of list experimentation. We develop a nonparametric estimator of the prevalence of sensitive behaviors that combines list experimentation and direct questioning. We prove that this estimator is asymptotically more efficient than the standard difference-in-means estimator, and we provide a basis for inference using Wald-type confidence intervals. Additionally, leveraging information from the direct questioning, we derive two nonparametric placebo tests of the identifying assumptions for the list experiment. We demonstrate the effectiveness of our combined estimator and placebo tests with an original survey experiment.

} \end{abstract}
\thispagestyle{empty}
\clearpage

\setcounter{page}{1}


\section{Introduction}

The prevalence of sensitive attitudes and behaviors is difficult to estimate using standard survey techniques due to the tendency of respondents to withhold information in such settings. In recent years, the list experiment has grown in popularity as a method for eliciting truthful responses to sensitive questions. Introduced as the ``item count technique'' by Miller (1984), the procedure has been used to study racial prejudice (Kuklinski, Cobb, and Gilens 1997a; Sniderman and Carmines 1997; Redlawsk, Tolbert, and Franko 2010), drug use (Biemer et al. 2005; Coutts and Jann 2011), risky sexual activity (LaBrie and Earleywine 2000; Walsh and Braithwaite 2008), vote buying (Gonzalez-Ocantos et al. 2012), and support for NATO-led forces among Pashtun respondents (Blair, Imai, and Lyall 2013). The standard list experiment proceeds by randomly partitioning respondents into control and treatment groups.  Subjects in the control group receive a list of $J$ non-sensitive items and report how many of the items apply to them. Subjects in the treatment group receive a list of $J+1$ items comprised of the same $J$ non-sensitive items plus one sensitive item. The list experiment estimate of the prevalence of the sensitive behavior is the difference-in-means between the treatment and control groups. The list experiment gives respondents cover to admit to engaging in the sensitive behavior -- so long as the respondent reports between $1$ and $J$ items, the researcher ostensibly cannot know whether an individual respondent engages in the sensitive behavior, but aggregate prevalence can be estimated. 

List experiments may be useful because prevalence estimates based on direct questions are biased when some subjects tell the truth and others withhold. In particular, the researcher cannot distinguish a respondent who does not engage in the sensitive behavior from one who does but is withholding: both types answer ``No'' to the direct question. Nevertheless, direct questions provide an important source of information when subjects \textit{admit} to engaging in a sensitive behavior. 
Direct questions are biased but yield precise estimates of prevalence. Under some assumptions, list experiments provide unbiased estimates of prevalence, but these estimates can be quite variable. The method we detail below allows researchers to reap the benefits of both direct questions and list experiments: increased precision and decreased bias. The central intuition of our approach is that, given a Monotonicity assumption (no false confessions), the true prevalence is a weighted average of two subject types: those who admit to the sensitive behavior and those who withhold; we estimate the former with direct questions and the latter with list experiments.

A popular design for the list experiment is to randomly split the sample into three groups: those receiving the control list and no direct question, those receiving the treatment list and no direct question, and those receiving a direct question but no list at all (Brueckner, Morning, and Nelson 2005; Holbrook and Krosnick 2010; Heerwig and McCabe 2009). This design is often used so that direct and list experiment estimates can be compared within the same population. A variant of this design asks only subjects in the control group the direct question (Ahart and Sackett 2004; Gilens, Sniderman, and Kuklinski 1998). Our estimator requires that both treatment and control subjects receive a direct question. Examples of this more extensive measurement approach include Droitcour et al. (1991) and Gonzales-Ocantos et al. (2012). Echoing similar design advice given in Kramon and Weghorst (2012) and Blair and Imai (2012), we advocate asking the direct question whenever feasible and investigating any possible ordering effects.

When respondents are asked both direct and list questions, researchers can also test the core assumptions underlying the list experiment: No Liars, No Design Effects, and ignorable treatment assignment (as found in Imai 2011). The No Liars assumption requires that those who engage in the sensitive behavior do in fact include the sensitive item when reporting the number of list items that apply. The No Design Effects assumption requires that subjects' responses to the non-sensitive items on the list are unaffected by the presence or absence of the additional sensitive item. These tests complement the one proposed by Blair and Imai (2012), which assesses whether any identified proportions of respondent types are negative, which would imply a contradiction between the model and the observed data.

We propose two tests. The logic of the first test, which is formalized below, is as follows: under the core list experiment assumptions and a Monotonicity assumption, the treatment versus control difference-in-means is in expectation equal to 1 among those who answer ``Yes'' to the direct question. Failing to reject the null hypothesis that the true difference in means for this subset is equal to 1 is equivalent to failing to reject the null hypothesis that the assumptions hold. We also propose a test of a variant of the ignorable treatment assignment assumption by assessing the dependence between responses to the direct question and the experimental treatment. While not conclusively demonstrating that the assumptions hold, these test results may give researchers more confidence that their survey instruments are providing reliable prevalence estimates. 

Previous methodological work on list experiments has largely been focused on two goals: decreasing the variance of list experiment estimates and modeling prevalence in a multivariate setting. Droitcour et al. (1991) propose the `Double List Experiment'' design in which the prevalence of the same sensitive item is investigated by two list experiments conducted with the same subjects, thereby reducing sampling variability. Holbrook and Krosnick (2010) use multivariate regression with treatment-by-covariate interaction terms to explore prevalence heterogeneity. Glynn (2013) suggests constructing the non-sensitive
items so that they are negatively
correlated with one another, a 
design feature that simultaneously reduces baseline variability and avoids ceiling effects. Corstange (2009) modifies the standard list experiment design by asking the control group each of the non-sensitive items directly, so that response to the non-sensitive items can be modeled and more precise estimates of the sensitive items can be calculated. Imai (2011) proposes a nonlinear least squares estimator and a maximum likelihood estimator to model responses with covariate data. Blair and Imai (2012) offer a detailed review of these techniques.   

Our contribution to the list experiment literature is to show the ease with which the additional information yielded by direct questioning can be incorporated into existing techniques. We demonstrate our proposed estimator and placebo tests on data from an original survey experiment conducted on Amazon's Mechanical Turk platform.  We conclude with suggestions for the design and analysis of list experiments in scenarios where it is ethically feasible to ask direct questions as well.

\section{Setting and Identification}

Suppose we have a random sample of $n$ subjects independently drawn from a large population. Let $X_i=1$ if subject $i$ engages in a sensitive behavior and $X_i=0$ otherwise.  We attempt to measure the behavior $X_i$ using two methods: direct questioning and list experimentation. Our goal is to identify the prevalence of the sensitive behavior in the population, $\mu = \Pr[X_i = 1]$.  Let $Y_i$ be the report of subject $i$ to the direct question.  We assume that, under direct questioning, subjects may lie and claim that they do not engage in the behavior but will not lie and falsely claim that they do engage in the behavior. 

\begin{assumption}[Monotonicity] There exist three latent classes of respondents under direct questioning: those who do not engage in the behavior and report truthfully ($X_i=0$, $Y_i=0$), subjects who engage and report truthfully ($X_i=1$, $Y_i=1$), and subjects who engage but report that they do not, i.e., withhold ($X_i=1$, $Y_i=0$).
\end{assumption} 

Let $p=\Pr[Y_i=1|X_i=1]$ be the probability of a subject reporting truthfully to the direct question, given that he or she engages in the sensitive behavior.  Then the response of subject $i$ is
\[ Y_i = \begin{cases} 0 & \text{with probability } 1-\mu + \mu(1-p) \\
                       1 & \text{with probability } \mu p. \end{cases} \]
The response $Y_i=0$ can be seen as a mixture of truthful negative reports and withholding.  
The probability that subject $i$ engages in the behavior, given a negative response, is therefore
\[ \Pr[X_i=1|Y_i=0] = \frac{(1-p)\mu}{1-\mu p}. \]
While direct questioning is sufficient to reveal $\Pr[Y_i = 1]=\mu p$, it is not sufficient to identify $\Pr[Y_i = 0|X_i=1]=1-p$.  

In contrast, the list experiment provides sufficient information to identify $\mu$. Suppose we have a treatment $Z_i \in \{0,1\}$.  In the list experiment, treated subjects ($Z_i=1$) receive a number of control questions and an additional question about the sensitive behavior.  We denote the number of items that the subject states are applicable with $V_i$.

\begin{assumption}[No Liars and No Design Effects]
Reframing Imai (2011)'s formulations, we observe $V_i = W_i + X_i Z_i$, where $W_i$ is the baseline outcome (under control) for subject $i$ for the list experiment. 
\end{assumption}

\noindent We further require that the treatment assignment be independent of the actual behavior, direct question, and baseline response. This is a stricter variant of Imai (2011)Õs ignorability assumption.

\begin{assumption}[Treatment Independence]
$(W_i, X_i, Y_i) \independent Z_i$.
\end{assumption}

\noindent Assumption 3 would be violated if (i) we did not have random assignment of the treatment or (ii) there are additional design effects; e.g., the treatment assignment affects the response to the direct question. Given random assignment, only the latter is a concern.

To ensure that all target quantities are well-defined (and, later, to facilitate inference), we impose the mild regularity condition that all population variances are positive.

\begin{assumption}[Non-degenerate Distributions]
$\Var[V_i | Z_i = z,Y_i = y] > 0$, for $z,y \in \{0,1\}$, $\Var[Z_i] > 0$ and $\Var[Y_i] > 0$.
\end{assumption}


We now turn to our primary identification result.

\begin{lemma}\label{id}
Given Assumptions 1-4 (Monotonicity, No Liars, No Design Effects, Treatment Independence, and Non-degenerate Distributions), the prevalence may be represented as
\begin{equation}\label{identification}
\mu = \E[Y_i] + \E[1 - Y_i] \left(\E[V_i | Z_i = 1, Y_i = 0]  - \E[V_i | Z_i = 0, Y_i = 0]\right).
\end{equation}
\end{lemma}
\begin{proof}
By Assumptions 2 and 3, 
\begin{equation}
    \E[V_i|Z_i=1] - \E[V_i|Z_i=0] = W_i + \mu - W_i = \mu 
  \label{eq:diff}
\end{equation}
Then expanding the left-hand side of \eqref{eq:diff} by marginalizing over $Y_i$, we represent the prevalence of the sensitive behavior as
\begin{equation*}
  \begin{split}
    \mu &= \E[V_i|Y_i=0,Z_i=1] \E[1-Y_i] + \E[V_i|Y_i=1,Z_i=1] \E[Y_i] \\
        &\quad - (\E[V_i|Y_i=0,Z_i=0] \E[1-Y_i] + \E[V_i|Y_i=1,Z_i=0] \E[Y_i]) \\
        &= (\E[V_i|Y_i=1,Z_i=1] - \E[V_i|Y_i=1,Z_i=0] ) \E[Y_i]  \\
        &\quad + (\E[V_i|Y_i=0,Z_i=1] - \E[V_i|Y_i=0,Z_i=0] ) \E[1-Y_i] \\ 
                &= \E[X_i|Y_i=1] \E[Y_i]  + (\E[V_i|Y_i=0,Z_i=1] - \E[V_i|Y_i=0,Z_i=0] ) \E[1-Y_i].
  \end{split}
\end{equation*}
The result follows since 
$ \E[X_i|Y_i=1] = 1 $
by Assumption 2.
\end{proof}
\noindent Note that if Assumptions 2-4 hold, but Assumption 1 (Monotonicity) does not hold, then  $\E[Y_i] + \E[1 - Y_i] \left(\E[V_i | Z_i = 1, Y_i = 0]  - \E[V_i | Z_i = 0, Y_i = 0]\right) > \mu$, as then $\E[X_i|Y_i = 1] < 1$. We now compare our identification result to standard results for the list experiment.
\begin{lemma}\label{id2}
Given Assumptions 2-4, 
\begin{equation}\label{identification2}
\mu = \E[V_i | Z_i = 1]  - \E[V_i | Z_i = 0].
\end{equation}
\end{lemma}
\noindent Lemma \ref{id2} has been proven by, e.g., Imai (2011, p. 409).

\section{Estimation, Inference and Efficiency}

In this section, we propose a simple nonparametric estimator of $\mu$ based on (1) and provide a basis for inference using Wald-type confidence intervals under a normal approximation. We also prove that our estimator is asymptotically more efficient than the standard difference-in-means estimator for the list experiment alone.

Define the sample means 
\[ \overline{Y} = \frac{1}{n}\sum_{i=1}^n Y_i \]
\[ \overline{V}_{1,0} = \frac{\displaystyle \sum_{i=1}^n V_i\ \indicator{Z_i=1,Y_i=0}}{ \displaystyle \sum_{i=1}^n \indicator{Z_i=1,Y_i=0}} 
  \qquad\text{and}\qquad
 \overline{V}_{0,0} = \frac{\displaystyle \sum_{i=1}^n V_i\ \indicator{Z_i=0,Y_i=0}}{ \displaystyle \sum_{i=1}^n \indicator{Z_i=0,Y_i=0}}  \]
where $\indicator{\cdot}$ is the indicator function.  Define an estimator of $\mu$ based on \eqref{identification},
\[ \hat\mu = \overline{Y} + (1 - \overline{Y}) \left(\overline{V}_{1,0}  - \overline{V}_{0,0}\right). \]

We can derive results on the sampling variance of $\hat\mu$.  Let $\gamma=\Pr(Z_i=1)$ be the probability of receiving the treatment question in the list experiment.
\begin{prop}\label{variance}
Given Assumption 4 (Non-degenerate Distributions), the asymptotic variance of $\hat\mu$ is characterized by
\begin{align}
    \underset{n\rightarrow \infty}{\plim}  n\Var \left[ \hat\mu  \right] = \frac{(1-\mu)^2}{1-\mu p}  \mu p + (1-\mu p) \Bigg[ \frac{\Var[V_i | Z_i = 1, Y_i = 0]}{\gamma} + \frac{\Var[V_i | Z_i = 0, Y_i = 0]}{1-\gamma}   \Bigg].
  \label{asyvar}
\end{align}
\end{prop}
\noindent Proofs for Propositions 1-3 are given in Appendix B. Under Assumptions 1-4, $\hat\mu$ is root-$n$ consistent and asymptotically normal, with a consistent estimator of the variance obtained by substituting sample analogues (i.e., sample means and sample variances) for population quantities. Namely, let 
$$
\widehat \Var\left[\hat\mu\right] = \frac{(1-\hat \mu)^2}{1- \hat \mu p}  \hat \mu \hat p + (1- \hat \mu \hat p) \Bigg[ \frac{\hat \sigma^2(V_i | Z_i = 1, Y_i = 0)}{\hat \gamma} + \frac{\hat \sigma^2(V_i | Z_i = 0, Y_i = 0)}{1-\hat \gamma}   \Bigg],
$$
where $\hat \sigma^2 (\cdot)$ denotes the sample variance, $\hat \gamma = \sum_{i=1}^n Z_i/n$, and \[\hat p = \frac{\displaystyle \sum_{i=1}^n Y_i X_i}{ \displaystyle \sum_{i=1}^n X_i}.\]
These properties are sufficient for construction of Wald-type confidence intervals using $\widehat \Var\left[\hat\mu\right]$.

\begin{prop}\label{confidenceintervals}
If Assumptions 1-4 hold, then confidence intervals constructed as $\hat\mu \pm z_{1-\alpha/2} \sqrt { \widehat \Var\left[\hat\mu\right] }$ will have $\mu$ $100(1-\alpha)\%$ coverage for $\mu$ for large $n$.
\end{prop}


An estimator based upon \eqref{identification} will have efficiency gains relative to standard estimators based upon \eqref{identification2}.  Consider the standard difference-in-means-based estimator for the list experiment,
\[ \hat\mu_S = \overline{V}_1 - \overline{V}_0 \]
where 
\[ \overline{V}_1 = \frac{\sum_{i=1}^n V_i Z_i}{\sum_{i=1}^n Z_i} \qquad\text{and}\qquad \overline{V}_0 = \frac{\sum_{i=1}^n V_i (1- Z_i)}{\sum_{i=1}^n (1- Z_i)} \]
We now show that the combined estimator $\hat\mu$ is asymptotically more precise than $\hat\mu_S$.
\begin{prop}\label{efficient}
Under Assumptions 1-4 (Monotonicity, No Liars, No Design Effects, Treatment Independence, and Non-degenerate Distributions),
\[ \underset{n\rightarrow\infty}{\plim} n\Var \left[ \hat\mu  \right]  < 
  \underset{n\rightarrow\infty}{\plim} n\Var \left[ \hat\mu_S  \right].  \]
\end{prop}

\section{Placebo Tests}

In this section, we derive two placebo tests to assess the validity of the identifying assumptions.

\subsection{Placebo Test I}

\label{sec:firstplacebo}
It is possible to jointly test the Monotonicity, No Liars, No Design Effects, and Treatment Independence assumptions. Under these assumptions, for all $t$, $\Pr[V_i = t | Z_i = 0, Y_i = 1] = \Pr[V_i = (t + 1) | Z_i = 1, Y_i = 1]$, thus tests of distributional equality are appropriate. Any valid test of distributional equality between $V_i$ (under $Z_i = 0, Y_i = 1$) and $V_i + 1$ (under $Z_i = 1, Y_i = 1$) will permit rejection of the null. 

However, since distributional equality implies that $\E[V_i | Z_i = 1, Y_i = 1] -  \E[V_i | Z_i = 0, Y_i = 1] = 1$, a simple test is available. Define $\beta = \E[V_i | Z_i = 1, Y_i = 1] -  \E[V_i | Z_i = 0, Y_i = 1]$. Consider estimators
\[ \hat\beta = \overline{V}_{1,1} - \overline{V}_{0,1} \]
and  
\[ \widehat\Var[\hat\beta] =  \frac{\widehat \sigma^2(V_i | Z_i = 1,Y_i = 1)}{\sum_{i=1}^n \indicator{Z_i = 1,Y_i = 1}} + \frac{\widehat\sigma^2(V_i | Z_i = 0,Y_i = 1)}{\sum_{i=1}^n \indicator{Z_i = 0,Y_i = 1}}. \]

\begin{prop}\label{placebo}
Under the null hypothesis that Assumptions 1-4 (Monotonicity, No Liars, No Design Effects, and Treatment Independence) hold, $\beta = 1$. For large $n$, if Assumption 4 (Non-degenerate Distributions) holds, then a two-sided $p$-value is given by 
$$
2 \Phi\left(-{\vert \hat \beta - 1 \vert} / {\sqrt{\widehat\Var[\hat\beta]}}\right),
$$
where $\Phi(.)$ is the normal CDF.
\end{prop}
\noindent A proof for Proposition \ref{placebo} follows directly from calculations analogous to those for Proposition \ref{confidenceintervals}. 

We also explore the power of Placebo Test I using a series of Monte Carlo simulations. We vary a number of factors, including: the number of subjects answering ``Yes'' to the direct question, the proportions of false confessors, liars, and the design-affected, and the variance of responses to the control list. These results are presented in Appendix C.

\subsection{Placebo Test II}
\label{sec:secondplacebo}
We can probe the validity of the Treatment Independence assumption with a second placebo test. If the answer to the direct question is statistically nonindependent of the treatment assignment (i.e, $Y_i \not\independent Z_i$), then it must be the case that the Treatment Independence assumption is invalid. Define $\delta = \E[Y_i | Z_i = 1] - \E[Y_i | Z_i = 0]$. Consider the estimators
$$\hat\delta =  \frac{\sum_{i=1}^n Y_i Z_i}{\sum_{i=1}^nZ_i} - \frac{\sum_{i=1}^n Y_i (1-Z_i)}{\sum_{i=1}^n (1-Z_i)}$$
and
%

\[ \widehat\Var[\hat\delta] =  \frac{\widehat \sigma^2(Y_i | Z_i = 1)}{\sum_{i=1}^n Z_i} + \frac{\widehat \sigma^2(Y_i | Z_i = 0)}{\sum_{i=1}^n (1-Z_i)}. \]

\begin{prop}\label{placebo2}
Under the null hypothesis that Assumption 3 holds, $\delta = 0$. For large $n$, if Assumption 4 (Non-degenerate Distributions) holds, then a two-sided $p$-value is given by 
$$
2 \Phi\left(-{\vert \hat \delta \vert} / {\sqrt{\widehat\Var[\hat\delta]}}\right).
$$
\end{prop}
\noindent A proof for Proposition \ref{placebo2} again follows directly from calculations analogous to those for Proposition \ref{confidenceintervals}. When the treatment is randomly assigned, then Placebo Test II is a simply a test of whether $Z_i$ has a causal effect on $Y_i$. Accordingly, under random assignment of $Z_i$, if treatment assignment follows the administration of the direct question, then Assumption 3 must hold.



\section{Application}
\label{sec:application}
We tested the properties of our estimator with a pair of studies carried out on Amazon's Mechanical Turk service, an internet platform where subjects perform a wide variety of tasks in return for compensation. Our main purpose was to assess the properties of our estimator with experimental data using an array of different
behaviors, some of which may
be considered socially sensitive. The relative anonymity of internet surveys provides a favorable environment for list experiments precisely because we expect subjects to withhold less often than they might in face-to-face or telephone settings. 

\subsection{Experimental Design}
The two studies were each comprised of the same five list experiments associated with five direct questions. The exact wording of the list experiments and direct questions is given in Appendix A. In three list experiments, we chose topics that are not socially sensitive: preferences over alternative energy sources, neighborhood characteristics, and news organizations. Two of the five list experiments dealt with racial and religious prejudice, topics where we would expect some withholding of anti-Hispanic and anti-Muslim sentiment. 

We recruited a convenience sample of 1,023 subjects from Mechanical Turk. We offered subjects \$1.00 to complete our survey, which is equivalent to a \$15.45 hourly rate -- a comparatively high wage by the standards of Mechanical Turk (Berinsky, Huber, and Lenz 2012). In order to defend against the potential for subjects to supply answers without reading or considering our questions, we included an ``attention question'' that required subjects to select a particular response in order to continue with the survey. Two subjects failed this quality check, and we exclude them from the main analysis. An additional seven subjects failed to respond to one or more of our questions, so we exclude them from the main analysis as well. The resulting sample size is $n = 1,014$.

Subjects were randomly assigned to either Study A or Study B with equal probability $0.5$. In Study A, direct questions were posed before the list questions, whereas in Study B, list questions were asked first. Subjects in both studies were then assigned to either the treatment or control conditions of each of the five list experiments. Table~\ref{tab:nsubjects} displays the number of subjects in each treatment condition for each study, as well as every pairwise crossing of conditions. All randomizations used Bernoulli random assignment with equal probability $0.5$. Consistent with our randomization procedure, each cell in the table (with the exception of the diagonal) contains approximately one-quarter of the subjects.

Before being randomized into treatment groups, subjects answered a series of background demographic questions. Table~\ref{tab:covbalance} shows balance statistics across age, gender, political ideology, education, and race across the treatment and control groups for the first list experiment in both studies. Our subject pool is more likely to be white, male, liberal, well-educated, and young than the general population. This pattern is consistent with the demographic description of Mechanical Turk survey respondents given by Mason and Suri (2011).   

\begin{table}[H]
\centering
\singlespacing
\caption{Number of Subjects in Each Treatment Condition} 
\begin{tabular}{cccccccccccccc}\\
  \hline 
& & \multicolumn{2}{c}{Study A or B} & \multicolumn{2}{c}{List 1}  &\multicolumn{2}{c}{List 2} & \multicolumn{2}{c}{List 3} & \multicolumn{2}{c}{List 4}  &\multicolumn{2}{c}{List 5}\\
& & A & B & T & C & T & C & T & C & T & C & T & C\\ 
 \cmidrule(r){3-4} \cmidrule(r){5-6} \cmidrule(r){7-8}  \cmidrule(r){9-10} \cmidrule(r){11-12} \cmidrule(r){13-14}
 \multirow{2}{*}{Study A or B}
  & A &  500 & 0 & 255 & 245  & 262 & 238 & 225 & 275 & 257 & 243 & 268 & 232  \\ 
 & B &  0  &  514 & 242 & 272 & 255 & 259  & 261 & 253 & 243 & 271 & 254 & 260  \\ 
  \multirow{2}{*}{List 1}
 & T &  &  & 517 &   0 & 251 & 266 & 273 & 244 & 269 & 248 & 273 & 244 \\ 
 & C &  &  &   0 & 497 & 246 & 251 & 255 & 242 & 245 & 252 & 219 & 278 \\ 
   \multirow{2}{*}{List 2}
 & T &  &  &  &  & 497 &   0 & 253 & 244 & 246 & 251 & 240 & 257 \\ 
 & C &  &  &  &  &   0 & 517 & 275 & 242 & 268 & 249 & 252 & 265 \\ 
   \multirow{2}{*}{List 3}
 & T &  &  &  &  &  &  & 528 &   0 & 269 & 259 & 263 & 265 \\ 
 & C &  &  &  &  &  &  &   0 & 486 & 245 & 241 & 229 & 257 \\ 
   \multirow{2}{*}{List 4}
 & T &  &  &  &  &  &  &  &  & 514 &   0 & 256 & 258 \\ 
 & C &  &  &  &  &  &  &  &  &   0 & 500 & 236 & 264 \\ 
   \multirow{2}{*}{List 5}
 & T &  &  &  &  &  &  &  &  &  &  & 492 &   0 \\ 
 & C &  &  &  &  &  &  &  &  &  &  &   0 & 522 \\ 
   \hline
\end{tabular}
\label{tab:nsubjects} 
\end{table}

\begin{table}[H]
\centering
\singlespacing
\caption{Covariate Balance: List Experiment 1}
\begin{tabular}{rrrrr}
  \hline
  & \multicolumn{2}{c}{Study A} & \multicolumn{2}{c}{Study B}\\
       \cmidrule(r){2-3} \cmidrule(r){4-5} 
 & \multicolumn{1}{c}{Treat} & \multicolumn{1}{c}{Control} & \multicolumn{1}{c}{Treat} & \multicolumn{1}{c}{Control} \\ 
  \hline
18 to 24 & 24.49 & 25.49 & 23.53 & 32.64 \\ 
  25 to 34 & 40.82 & 43.92 & 41.18 & 40.91 \\ 
  35 to 44 & 20.00 & 15.29 & 17.28 & 14.88 \\ 
  45 to 54 & 8.16 & 10.98 & 10.66 & 6.20 \\ 
  55 to 64 & 4.90 & 3.92 & 5.88 & 3.72 \\ 
  65 or over & 1.63 & 0.39 & 1.47 & 1.65 \\ \\
  Female & 46.12 & 47.45 & 46.32 & 42.98 \\ 
  Male & 53.88 & 52.16 & 53.31 & 57.02 \\ 
  Prefer not to say -- Gender & 0.00 & 0.39 & 0.37 & 0.00 \\ \\
  Liberal & 50.61 & 42.35 & 50.37 & 51.65 \\ 
  Moderate & 28.16 & 34.51 & 26.84 & 28.51 \\ 
  Conservative & 18.37 & 19.22 & 20.59 & 14.46 \\ 
  Haven't thought much about this & 2.86 & 3.92 & 2.21 & 5.37 \\ \\
  Less than High School & 0.41 & 0.78 & 0.74 & 1.24 \\ 
  High School / GED & 11.02 & 11.76 & 10.66 & 7.44 \\ 
  Some College & 42.45 & 42.75 & 40.44 & 40.91 \\ 
  4-year College Degree & 33.47 & 33.73 & 35.29 & 36.36 \\ 
  Graduate School & 12.65 & 10.98 & 12.87 & 14.05 \\  \\
  White, non Hispanic & 79.18 & 77.65 & 79.78 & 80.17 \\ 
  African-American & 7.76 & 10.98 & 6.62 & 4.55 \\ 
  Asian/Pacific Islander & 5.71 & 4.71 & 5.15 & 9.92 \\ 
  Hispanic & 4.90 & 3.92 & 6.25 & 4.13 \\ 
  Native American & 0.82 & 0.78 & 1.10 & 0.83 \\ 
  Other & 1.22 & 1.18 & 1.10 & 0.41 \\ 
  Prefer not to say -- Race & 0.41 & 0.78 & 0.00 & 0.00 \\ \\
 $n$ & 245 & 255 & 272 & 242 \\ 
   \hline
\end{tabular}
\label{tab:covbalance} 
\end{table}
\newpage
\subsection{Study A (Direct Questions First)}

In Study A, subjects were presented with a battery of the five direct questions before receiving the five list experiments. Table~\ref{tab:studyAests} presents three estimates of the prevalence in our subject pool. The first is a naive estimate computed by taking the average response to the direct question, $\overline{Y}$ (Direct). The remaining two estimates are $\hat\mu_S$ (Standard List) and $\hat\mu$ (Combined List). For example, the direct question estimate  $\overline{Y}$ of the percentage agreeing that Muslims should not be allowed to teach in public schools is 11\%, the list experiment estimate  $\hat\mu_S$ is 17\%, and the combined estimate $\hat\mu$  is 19\%. Of particular note are the standard errors associated with the standard list experiment as compared with those associated with the combined estimate: the reductions in estimated sampling variability are dramatic, ranging from 14\% to 67\%. As expected, reductions 
tend to be
larger when a larger number of subjects respond ``Yes'' to the direct question. Figure~\ref{fig:studyAfig} presents these results graphically: the estimates generally agree (providing confidence that the list experiments and the direct questions are measuring the same quantities), and the 95\% confidence intervals around the combined estimate are always tighter than those around the standard estimate.

\begin{table}[H]
\centering
\singlespacing
\caption{Study A (Direct First): Three Estimates of Prevalence}
\begin{tabular}{rp{1cm}p{1cm}p{1cm}p{1cm}p{1cm}p{1cm}c}
  \hline  
    \noalign{\vskip 1mm} 
   & \multicolumn{2}{c}{Direct} & \multicolumn{2}{c}{Standard List} & \multicolumn{2}{c}{Combined List} &  \multirow{2}{*}{\parbox{25mm}{\centering \% Reduction in Sampling Variance}}\\
         \noalign{\vskip 1mm} 
    \cmidrule(r){2-3} \cmidrule(r){4-5} \cmidrule(r){6-7}
  & \multicolumn{1}{c}{$\overline{Y}$} & \multicolumn{1}{c}{SE} & \multicolumn{1}{c}{$\hat\mu_S$} & \multicolumn{1}{c}{SE}& \multicolumn{1}{c}{$\hat\mu$} & \multicolumn{1}{c}{SE} &  \\
   \noalign{\vskip 1mm}  
  \hline
Nuclear Power & 0.656 & 0.021 & 0.748 & 0.084 & 0.666 & 0.049 & 66.8 \\ 
  Public Transportation & 0.538 & 0.022 & 0.513 & 0.072 & 0.627 & 0.049 & 54.0 \\ 
  Spanish-speaking & 0.102 & 0.014 & 0.035 & 0.079 & 0.042 & 0.074 & 14.0 \\ 
  Muslim Teachers & 0.110 & 0.014 & 0.166 & 0.081 & 0.187 & 0.074 & 15.4 \\ 
  CNN & 0.444 & 0.022 & 0.338 & 0.105 & 0.533 & 0.070 & 55.3 \\ 
   \hline
      \multicolumn{8}{l}{$n$ = 500 for all estimates}
\end{tabular}
\label{tab:studyAests} 
\end{table}

\begin{figure}[H]
\centering
\caption{Study A (Directs First): Three Estimates of Prevalence}
\includegraphics[scale=.45]{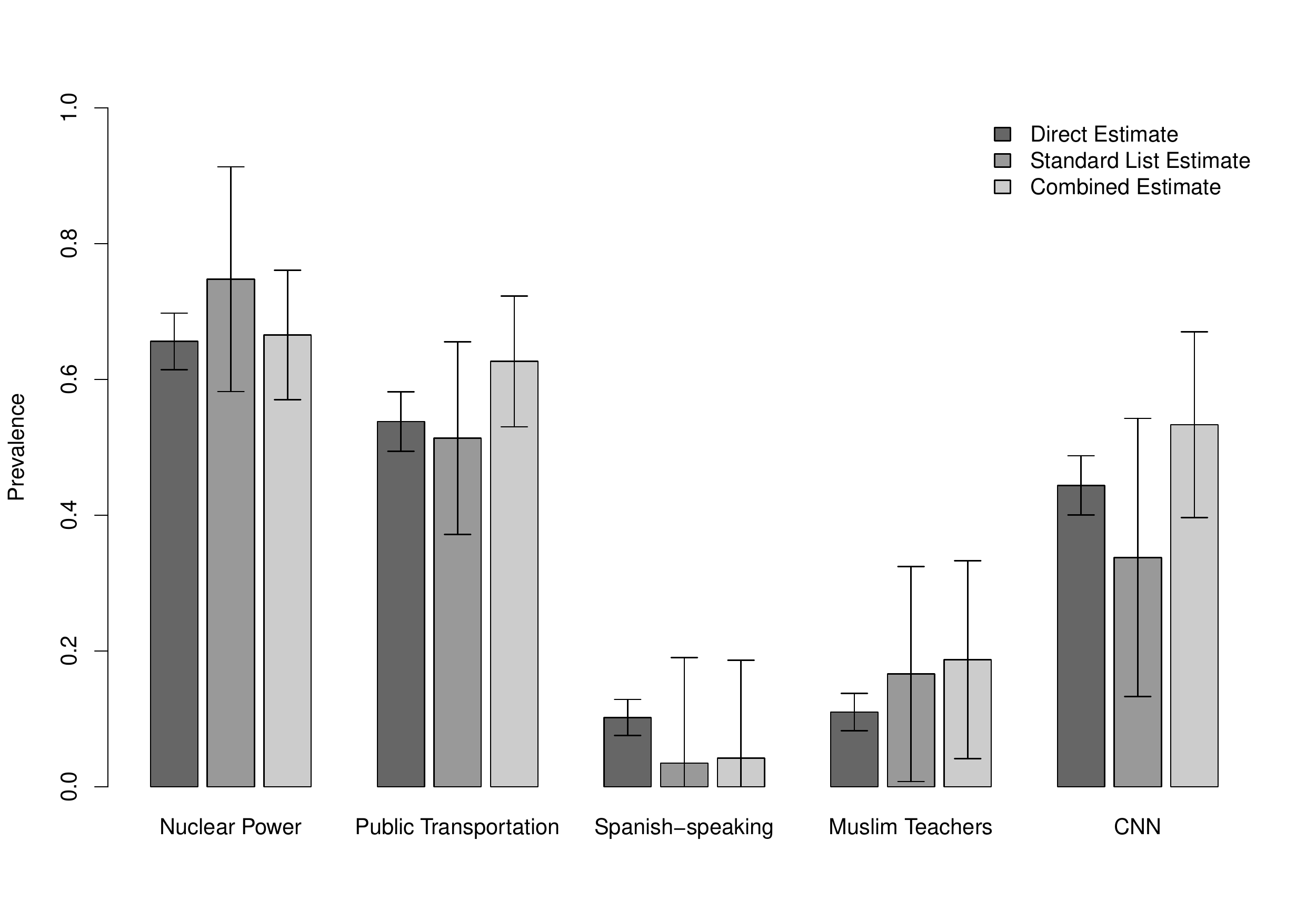}
\label{fig:studyAfig}
\end{figure}

Table~\ref{tab:studyAplacebo} presents the results of Placebo Test I. If these assumptions hold, the standard list experiment difference-in-means estimator will recover estimates that are in expectation equal to one among the subsample that answers ``Yes'' to the direct question. In two cases, we reject the joint null hypothesis of Monotonicity, No Liars and No Design Effects: Public Transportation ($p=0.02$) and CNN ($p=0.03$).  

\begin{table}[H]
\centering
\singlespacing
\caption{Study A (Direct First): Placebo Test I} 
\begin{tabular}{rrrrr}
  \hline
      \noalign{\vskip 2pt}  
 & \multicolumn{1}{c}{$\hat \beta$} & \multicolumn{1}{c}{SE} &\multicolumn{1}{c}{$\Pr[\beta \neq 1]$} & \multicolumn{1}{c}{$n$} \\ 
  \hline
Nuclear Power & 1.054 & 0.095 & 0.568 & 328\\ 
  Public Transportation & 0.790 & 0.091 & 0.021 & 269\\ 
  Spanish-speaking & 0.848 & 0.279 & 0.585 & 51\\ 
  Muslim Teachers & 1.008 & 0.237 & 0.973 & 55 \\ 
  CNN & 0.696 & 0.143 & 0.034 & 222\\ 
   \hline
\end{tabular}
\label{tab:studyAplacebo}
\end{table}

Since we have random assignment and the experimental treatment follows the administration of the direct question, we expect to pass Placebo Test II. Indeed, Placebo Test II results show no significant differences in mean responses to the direct questions by the list experimental treatment assignments.

\begin{table}[H]
\centering
\caption{Study A (Directs First): Placebo Test II}
\begin{tabular}{rrr}
  \hline
 & Estimate & SE \\ 
  \hline
Nuclear Power & 0.066 & 0.043 \\ 
  Public Transportation & -0.000 & 0.045 \\ 
  Spanish-speaking & 0.016 & 0.027 \\ 
  Muslim Teachers & -0.030 & 0.028 \\ 
  CNN & -0.056 & 0.045 \\ 
   \hline

   \hline
\end{tabular}
\label{placebo2directsfirstmain}
\end{table}

\subsection{Study B (List Experiments First)}

Study B reverses the order of the direct questions and list experiments: subjects participated in all five list experiments before answering the direct questions. This design choice risks priming
subjects in the
treatment group
in ways that might
alter their responses
to subsequent direct questions.
For example, treated subjects
may be prone to misreport if subjects suspect that a particular topic is being given special scrutiny.

\begin{table}[H]
\centering
\singlespacing
\caption{Study B (Lists First): Three Estimates of Prevalence}
\begin{tabular}{rp{1cm}p{1cm}p{1cm}p{1cm}p{1cm}p{1cm}c}
  \hline  
    \noalign{\vskip 1mm}  
   & \multicolumn{2}{c}{Direct} & \multicolumn{2}{c}{Standard List} & \multicolumn{2}{c}{Combined List} &  \multirow{2}{*}{\parbox{25mm}{\centering \% Reduction in Sampling Variance}}\\
      \noalign{\vskip 1mm}  
    \cmidrule(r){2-3} \cmidrule(r){4-5} \cmidrule(r){6-7}
  & \multicolumn{1}{c}{$\overline{Y}$} & \multicolumn{1}{c}{SE} & \multicolumn{1}{c}{$\hat\mu_S$} & \multicolumn{1}{c}{SE}& \multicolumn{1}{c}{$\hat\mu$} & \multicolumn{1}{c}{SE} &  \\
   \noalign{\vskip 1mm}  
  \hline
Nuclear Power & 0.603 & 0.022 & 0.499 & 0.089 & 0.624 & 0.052 & 65.3 \\ 
  Public Transportation & 0.539 & 0.022 & 0.578 & 0.073 & 0.608 & 0.051 & 51.8 \\ 
  Spanish-speaking & 0.113 & 0.014 & 0.144 & 0.083 & 0.149 & 0.076 & 16.3 \\ 
  Muslim Teachers & 0.103 & 0.013 & 0.107 & 0.083 & 0.123 & 0.074 & 20.1 \\ 
  CNN & 0.496 & 0.022 & 0.645 & 0.101 & 0.587 & 0.065 & 59.4 \\ 
   \hline
      \multicolumn{8}{l}{$n$ = 514 for all estimates}
\end{tabular}
\label{tab:studyBests}
\end{table}

\begin{figure}[H]
\centering
\caption{Study B (Lists First): Three Estimates of Prevalence}
\includegraphics[scale=.45]{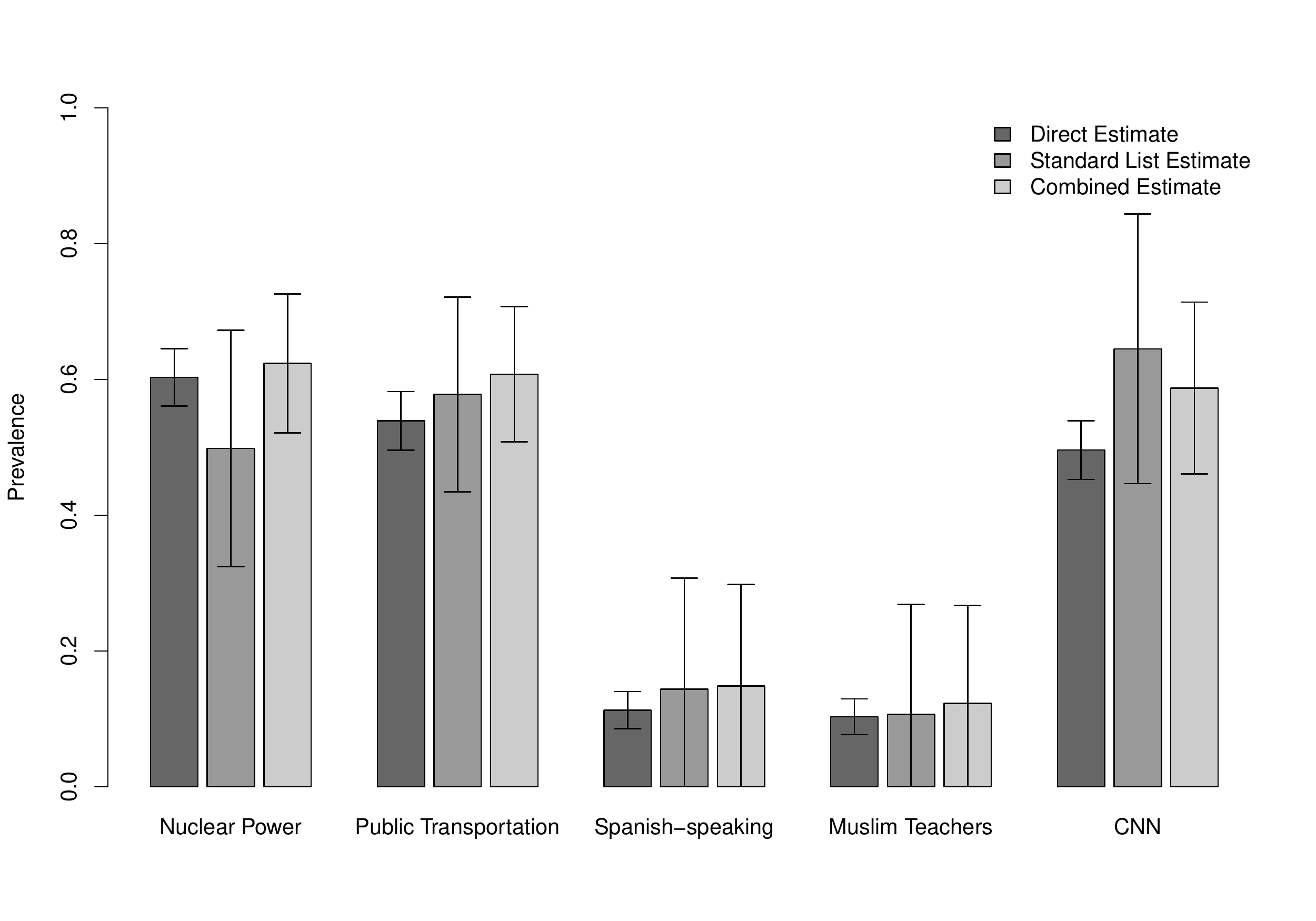}
\label{fig:studyBfig}
\end{figure}

All three prevalence estimates for Study B are presented in Table~\ref{tab:studyBests} and Figure~\ref{fig:studyBfig}. The direct question estimates are very similar between Study A and Study B -- none of the differences between the estimates is significant at the $0.05$ level. The standard list experiment estimates differ between Studies A and B, suggesting a question order effect. The combined estimator produces tighter estimates in Study B as well, with estimated sampling variability reductions in a very similar range. Appendix E presents formal tests of the differences in estimates across the studies.

The results of Placebo Test I for Study B are presented in Table~\ref{tab:studyBplacebo}. Among the subgroup of respondents who answer ``Yes'' to the direct question, the list experiment difference-in-means estimate $\beta$ should be equal to 1, under Assumptions 1-4 (Monotonicity, No Liars, No Design Effects, and Treatment Independence). None of the values of $\beta$ are statistically significantly different from 1 using the placebo test, and a joint test via Fisher's method is insignificant as well.

\begin{table}[H]
\centering
\singlespacing
\caption{Study B (Lists First): Placebo Test I}
\begin{tabular}{rrrrr}
  \hline
      \noalign{\vskip 2pt}  
 & \multicolumn{1}{c}{$\hat \beta$} & \multicolumn{1}{c}{SE} &\multicolumn{1}{c}{$\Pr[\beta \neq 1]$} & \multicolumn{1}{c}{$n$} \\ 
  \hline
Nuclear Power & 0.881 & 0.113 & 0.294 & 310\\ 
  Public Transportation & 0.913 & 0.091 & 0.339 & 277\\ 
  Spanish-speaking & 0.767 & 0.229 & 0.309 & 58\\ 
  Muslim Teachers & 0.700 & 0.285 & 0.293 & 53\\ 
  CNN & 0.847 & 0.135 & 0.256 & 255\\ 
   \hline
\end{tabular}
\label{tab:studyBplacebo}
\end{table}

As described in section \ref{sec:secondplacebo}, the combined estimator relies in part on the assumption that a subject's response to the direct question is unaffected by the list experimental treatment assignment. Violations of this assumption are directly testable using the second placebo test discussed in section~\ref{sec:secondplacebo}.  In study B, subjects were exposed to either a treatment or a control list before answering the direction question. Table~\ref{secondplacebomain} below presents the effect the treatment lists may have had on answers to the direct questions. In two of the five cases, direct questions were significantly affected by the treatment list:  Treated subjects were 8.6 percentage points less likely to declare their support for Nuclear power and were 13.2 percentage points more likely to report watching CNN.  These findings indicate that the Treatment Independence assumption is most likely violated for these questions, rendering the Study B combined list estimates for these two questions unreliable.

\begin{table}[H]
\centering
\caption{Study B (Lists First): Placebo Test II}
\begin{tabular}{rrr}
  \hline
 & Estimate & SE \\ 
  \hline
Nuclear Power & -0.086 & 0.043 \\ 
  Public Transportation & 0.034 & 0.044 \\ 
  Spanish-speaking & 0.027 & 0.028 \\ 
  Muslim Teachers & 0.016 & 0.027 \\ 
  CNN & 0.132 & 0.044 \\ 
   \hline

   \hline
\end{tabular}
\label{secondplacebomain}
\end{table}

\subsection{Replication Study}

We conducted a replication study following the identical design with 506 new Mechanical Turk subjects in a replication of Study A and 506 in a replication of Study B. Full results of this replication study are presented in Appendix F, but the findings are strikingly similar to the first investigation. The list experimental estimates diverge, though these differences are likely due to sampling variability (none of the differences in list experiment estimates are significantly different from one another).  One of five placebo tests was significant in Study A and three of five were significant in Study B.  Perhaps surprisingly, the effect of the treatment list on direct answers to the CNN question presented in Table~\ref{secondplacebomain} was also replicated.

\section{Discussion}

Social desirability effects may bias prevalence estimates of sensitive behaviors and opinions obtained using direct questioning, but that does not mean that direct questions are useless. Under an assumption of Monotonicity (subjects who do not engage in the sensitive behavior do not falsely confess), direct questions reveal reliable information about those who answer ``Yes.'' Among those who answer ``No,'' we cannot directly distinguish those who withhold from those who do not engage in the sensitive behavior -- for these subjects, list experiments are especially useful. Combining these two techniques into a single estimator yields more precise estimates of prevalence, and employing
direct and list
questions in tandem
also enables
the researcher
to test crucial
identifying 
assumptions.

A few caveats are in order with respect to empirical applications. First, Monotonicity is not guaranteed to hold, especially when social desirability cuts in opposite directions for different subgroups. For example, moderates in liberal areas may feel pressure to support Muslim teachers, whereas moderates in conservative areas may feel pressure to oppose them. Second, list experiments are often employed when the security of respondents would be compromised if they admitted to sensitive opinions or behaviors (e.g., Pashtun respondents admitting support for NATO forces, Blair, Imai, and Lyall 2013). We do not take these concerns lightly, and in such cases would not recommend the use of our method. Third, the order in which direct questions and list experiments are asked appears to matter: we reject the joint null hypothesis of Monotonicity, No Liars and No Design Effects for two of the experiments in Study A, but fail to reject it for any of the five experiments in Study B. Given that the two studies are roughly equally powered, this finding would seem to argue for asking the direct question second. However, Placebo Test II suggests that, at least in our application, asking the direct question second induced a violation of the Treatment Independence assumption. In sum, we recommend randomizing the order in which the list experiment and the direct question are presented, so that a) question-order effects can be contained and b) the relevant tests of the assumptions can be performed.

We have combined direct question estimates with the simplest of the various list experiment estimators: difference-in-means. This work could be extended straightforwardly to the multivariate settings discussed by Corstange (2009), Holbrook and Krosnick (2010), and Imai (2011). One such approach would involve regression estimation (S\"{a}rndal, Swensson, and Wretman 1992, Lin 2013) or, equivalently, post-stratification (Holt and Smith 1979, Miratrix, Sekhon and Yu 2013) for computing covariate adjusted means. Such an approach would  improve asymptotic efficiency without any parametric assumptions, and a consistent variance estimator may be derived by substituting residuals from the regression fit. Finally, we note that other methods for eliciting truthful responses to sensitive questions, such as randomized response (Warner 1965) and endorsement experiments (Bullock, Imai, and Shapiro 2011; Lyall, Blair, and Imai forthcoming), could also be profitably combined with direct questioning. 

\singlespacing
\nocite{*}
\bibliography{listbib}

\newpage
\setcounter{table}{0}
\setcounter{figure}{0}
\renewcommand{\thetable}{A\arabic{table}}
\renewcommand{\thefigure}{A\arabic{figure}}

\section*{Appendix A: Survey Text}
\singlespacing
\begin{longtable}{p{2cm} p{7cm} p{6cm}}
     & List Question (sensitive item in bold)  & Direct Question  \\ \hline
\endfirsthead

\multicolumn{3}{c}%
{continued from previous page} \\
     & List Question (sensitive item in bold) & Direct Question  \\ \hline
\endhead

\hline 
\endlastfoot 
     
Nuclear Power \newline (L1 \& D1) & Here is a list of four (five) alternative energy sources that some people think the United States should be developing.
Please read all four (five) and select how many of them you think the United States should be developing.
We do not want to know which ones, just how many. & Some people think nuclear power is an alternative energy source that the United States should be developing. 

How about you?  Do you think nuclear power is an alternative energy
source that the United States should be developing? \\*
          & Wind  &  Yes \\*
          & Solar  &  No \\*
          & \textbf{Nuclear Power } &  \\*
          & Natural Gas  &  \\*
          & Ethanol &  \\* \hline
Public Transportation \newline (L2 \& D2) & Here is a list of three (four) things some people consider important when choosing a neighborhood to live in.
Please read all three (four) and select how many of them you consider important when choosing a neighborhood.
We do not want to know which ones, just how many. & Some people consider access to public transportation important when choosing a
neighborhood to live in.  

How about you? Do you consider access to public transportation
important when choosing a neighborhood to live in? \\*
          & Proximity to shops and restaurants  &  Yes \\*
          & Quality of local schools  & No \\*
          & Neighbors who share my values &  \\*
          & \textbf{Access to public transportation} &  \\* \hline

Spanish-speaking   \newline (L3 \& D3)          & Here is a list of four (five) things that sometimes make people angry or upset.
Please read all four (five) and select how many of them upset you.
We do not want to know which ones, just how many. & Sometimes people are angry or upset when a Spanish-speaking family moves in next door. 
How about you? Would a Spanish-speaking family moving in next door upset you? \\*
          & The federal government increasing the tax on gasoline  & Yes \\*
          & Professional athletes earning large salaries  & No \\*
          & \textbf{A Spanish-speaking family moving in next door } &  \\*
          & Requiring seat belts be used when driving  &  \\*
          & Large corporations polluting the environment &  \\* \hline

Muslim Teachers  \newline (L4 \& D4)         & Here is a list of three (four) groups that some people think should be prohibited from teaching in public schools.
We do not want to know which ones, just how many.
 & Sometimes people think that Muslims should ​be prohibited from teaching in public schools. 
How about you?  Do you think that Muslims should ​be prohibited from teaching in public schools?
 \\*
          & 18-21 year olds & Yes \\*
          & \textbf{Muslims} & No \\*
          & People without a Masters degree in education &  \\*
          & People who earn a 2.0 GPA or lower  &  \\* \hline

CNN  \newline (L5 \& D5)         & Here is a list of four (five) news organizations.
Please read all four (five) and select how many you read or watch in the course of an ordinary month.
We do not want to know which ones, just how many. & In the course of an ordinary month, do you watch CNN? \\*
          & The New York Times  & Yes \\*
          & \textbf{CNN } & No \\*
          & The Huffington Post  &  \\*
          & Fox News  &  \\*
          & Politico &  \\     
\label{tab:svytext}       
\end{longtable}

\newpage
\section*{Appendix B: Proofs}

\subsection*{Proof of Proposition 1}
\begin{proof}
The proof proceeds by working with linearized variances. 
\begin{equation}\label{linear}
  \begin{split}
  \Var \left[ \hat\mu  \right] &= \Var\left[1 + (1-\overline{Y})\left(\overline{V}_{1,0} - \overline{V}_{0,0} - 1 \right) \right]  \\
                               &= \Var\left[(1-\overline{Y})\left(\overline{V}_{1,0} - \overline{V}_{0,0} - 1 \right) \right]  \\
                               &= (\E[1-\overline{Y}])^2 \Var\left[\overline{V}_{1,0} - \overline{V}_{0,0}\right]  +\Var[1-\overline{Y}] (\E\left[\overline{V}_{1,0} - \overline{V}_{0,0} - 1\right])^2  \\
                               &\qquad +\Var[1-\overline{Y}] \Var\left[\overline{V}_{1,0} - \overline{V}_{0,0}\right]  \\
                               &= (1-\mu p)^2 \Var\left[\overline{V}_{1,0} - \overline{V}_{0,0}\right] + \frac{\mu p(1-\mu p)}{n} \left(\frac{1-\mu}{1-\mu p}\right)^2 \\
                               &\qquad + \frac{\mu p(1-\mu p)}{n} \Var\left[\overline{V}_{1,0} - \overline{V}_{0,0}\right] \\
                               &= \frac{\mu p(1-\mu)^2}{n(1-\mu p)} + (1-\mu p)^2 \left[ \frac{\Var[V_i|Z_i=1,Y_i=0]}{(1-\mu p) m} + \frac{\Var[V_i|Z_i=0,Y_i=0]}{(1-\mu p) (n-m)} \right]  +O(n^{-2}) \\
                               &= \frac{\mu p(1-\mu)^2}{n(1-\mu p)} + (1-\mu p)\left[ \frac{\Var[V_i|Z_i=1,Y_i=0]}{m} + \frac{\Var[V_i|Z_i=0,Y_i=0]}{n-m} \right]  + O(n^{-2}), 
\end{split}
\end{equation}
where $m =  \sum_{i=1}^n Z_i$. Multiplying by $n$ and letting $\gamma=m/n$ yields the desired result.
\end{proof}

\subsection*{Proof of Proposition 2}
\begin{proof}
  Consistent estimators exist for each element in $\widehat\Var\left[\hat\mu\right]$ by substituting in sample means and sample variances. Since the random vector $\left( 1-\bar{Y}, \overline{V}_{1,0}, \overline{V}_{0,0} - 1\right)$ is jointly asymptotically normal, so too is $\hat\mu$. Since ${ ( \hat \mu - \mu )}/\sqrt { \Var\left[\hat\mu\right]} \rightarrow \N(0,1)$ and $\sqrt {\Var\left[\hat\mu\right]}/\sqrt { \widehat \Var\left[\hat\mu\right]} \rightarrow 1$, then Slutsky's Theorem implies $$\frac{ ( \hat \mu - \mu )}{\sqrt { \Var\left[\hat\mu\right]}} \frac{\sqrt {\Var\left[\hat\mu\right]}}{ \sqrt {\widehat \Var\left[\hat\mu\right]}} \rightarrow \N(0,1).$$
\end{proof}

\subsection*{Proof of Proposition 3}
\begin{proof}
We begin by expressing the asymptotic variance of $\hat \mu_S$:
\begin{align*}
\underset{n\rightarrow \infty}{\plim} n\Var \left[ \hat \mu_S  \right] =& (1-\mu p) \left[ \frac{\Var[V_i | Z_i = 1,Y_i = 0]}{\gamma}  + \frac{\Var[V_i | Z_i = 0, Y_i = 0]}{1-\gamma} \right] \\
&+ \mu p \left[ \frac{\Var[V_i | Z_i = 1,Y_i = 1]}{\gamma}  + \frac{\Var[V_i | Z_i = 0, Y_i = 1]}{1-\gamma} \right] \\
& + (1-\mu p)  \left[   \frac{ (\E[V_i | Z_i = 1,Y_i = 0] - \E[V_i | Z_i = 1])^2 }{\gamma}\right] \\  
& + (1-\mu p) \left[ \frac{(\E[V_i | Z_i = 0,Y_i = 0] - \E[V_i | Z_i = 0])^2}{1-\gamma} \right]\\
& + \mu p  \left[   \frac{ (\E[V_i | Z_i = 1,Y_i = 1] - \E[V_i | Z_i = 1])^2 }{\gamma}\right] \\  
& + \mu p \left[ \frac{(\E[V_i | Z_i = 0,Y_i = 1] - \E[V_i | Z_i = 0])^2}{1-\gamma} \right]\\
%
= & (1-\mu p) \left[ \frac{\Var[V_i | Z_i = 1,Y_i = 0]}{\gamma}  + \frac{\Var[V_i | Z_i = 0, Y_i = 0]}{1-\gamma} \right] \\
&+ \mu p \left[ \frac{\Var[V_i | Z_i = 1,Y_i = 1]}{\gamma}  + \frac{\Var[V_i | Z_i = 0, Y_i = 1]}{1-\gamma} \right] \\
& + \mu p(1-\mu p)\left[\frac{ ( \E[V_i | Z_i = 1,Y_i = 0]  -  \E[V_i | Z_i = 1,Y_i = 1]  )^2 }{\gamma}\right] \\  
& + \mu p(1-\mu p)  \left[ \frac{(\E[V_i | Z_i = 0,Y_i = 0]  - \E[V_i | Z_i = 0,Y_i = 1] )^2}{1-\gamma} \right].
\end{align*}
By Assumptions 1 and 2 (Monotonicity, No Liars and No Design Effects), 
$\E[V_i | Z_i = 1,Y_i = 1] = \E[V_i | Z_i = 0,Y_i = 1] + 1$. Then 
\begin{align*}
  \underset{n\to\infty}{\plim} n\Var \left[ \hat \mu_S  \right] =
&  (1-\mu p) \left[ \frac{\Var[V_i | Z_i = 1,Y_i = 0]}{\gamma}  + \frac{\Var[V_i | Z_i = 0, Y_i = 0]}{1-\gamma} \right] \\
&+ \mu p \left[ \frac{\Var[V_i | Z_i = 1,Y_i = 1]}{\gamma}  + \frac{\Var[V_i | Z_i = 0, Y_i = 1]}{1-\gamma} \right] \\
&+ \mu p(1-\mu p) \left[ \frac{ ( \frac{\mu -1}{1-\mu p}+  \E[V_i | Z_i = 0,Y_i = 0]  -  \E[V_i | Z_i = 0,Y_i = 1]  )^2 }{\gamma}\right] \\  
& + \mu p (1 - \mu p)  \left[ \frac{(\E[V_i | Z_i = 0,Y_i = 0]  - \E[V_i | Z_i = 0,Y_i = 1] )^2}{1-\gamma} \right].
\end{align*}
Applying the first order condition, $\plim_{n\rightarrow \infty} n\Var \left[ \hat\mu_S  \right]$ is minimized when $\E[V_i | Z_i = 0,Y_i = 0]  - \E[V_i | Z_i = 0,Y_i = 1]  = [\gamma-1][(\mu-1)/(1-\mu p)]$.
Substituting terms, it follows that 
\begin{align}\label{final}
\underset{n\rightarrow \infty}{\plim}
n \Var \left[  \hat \mu_S   \right] \geq & (1-\mu p) \left[ \frac{\Var[V_i | Z_i = 1,Y_i = 0]}{\gamma}  + \frac{\Var[V_i | Z_i = 0, Y_i = 0]}{1-\gamma} \right] \nonumber \\
&+ \mu p \left[ \frac{\Var[V_i | Z_i = 1,Y_i = 1]}{\gamma}  + \frac{\Var[V_i | Z_i = 0, Y_i = 1]}{1-\gamma} \right] \nonumber \\
& + \mu p  (1 - \mu p)  \left(\frac{1-\mu}{1-\mu p}\right)^2 \nonumber \\
\geq &  \underset{n\to\infty}{\plim} n\Var \left[  \hat \mu  \right].
\end{align}
Assumption 3 (Non-degenerate Distributions) ensures that the inequality holds strictly.
\end{proof}

\section*{Appendix C: Placebo Test I Power Analysis}

The placebo test proposed in Section 4 is a joint test of Assumptions 1-3 (Monotonicity, No Liars, No Design Effects, and Treatment Independence).  In brief, the test considers whether the conventional list experimental estimate appears to be significantly different from 1.0 among the subset of subjects who answer ``Yes'' to the direct question.  If this estimate is different from 1.0, it must either be because a) some of those who answer ``Yes'' are falsely confessing (thereby violating monotonicity) or b) the standard list experiment assumptions of No Liars and No Design Effects are not met.  We vary the following five quantities: the number of subjects who answer ``Yes'' to the direct question ($N_{\text{Yes}}$, or $\sum_{i=1}^N Y_i$) the variability of responses to the non-sensitive list items, the proportion of $N_{\text{Yes}}$ who falsely confess, the proportion of $N_{\text{Yes}}$ who lie when given the treatment list, and the proportion of $N_{\text{Yes}}$ whose responses to the non-sensitive list items change when given the treatment list. 

We display the results of four power simulations in Figure~\ref{fig: powertiles} below.  On the y-axis of each panel, we vary $N_{\text{Yes}}$.  On the x-axis of the first three panels, we vary the proportion of those subjects whose response profile violates one of the assumptions: No False Confessions, No Liars, or No Design Effects, respectively. The proportion of units that violate the other two assumptions was fixed at 0.  Control list responses were drawn from a binomial distribution with a success probability of 0.4 and four trials.  For those units who do meet the assumptions of No False Confessions, No Liars, or No Design Effects, treatment list responses were set equal to the control list response plus one. The treatment list responses for Liars and False Confessors were set equal to their control list responses. The treatment list response for Design Affected subjects was generated as a ``ceiling effect'': the control list plus one, except for those with a ``4'' on the control list; those units' treatment list response remained equal to 4.\footnote{This is one of many possible design effects; the power of our test to detect any particular design effect depends on the manner in which it changes subjects' treatment list responses.} The final panel fixes the proportion of false confessors at 0.20 and changes the variability of responses to the non-sensitive list items.  We parameterized the variability in responses to non-sensitive list items as the success probability of a binomial distribution with four trials.  This variability is maximized when $p = 0.5$.  

We varied $N_{\text{Yes}}$ from 100 to 1000 in steps of 50, the proportion of False Confessors, Liars, and Design Affected units from 0 to 1 in steps of 0.01, and the success probability of the binomial distribution from 0 to 1 in steps of 0.01.  We conducted 1,000 simulations of each combination.  The shading reflects the proportion of simulations in which we were able to reject the null hypothesis that Assumptions 1 and 2 hold, with darker shades corresponding to higher power.  For ease of interpretation, the shading around the conventional power target of 0.80 is shaded red.

The simulations show that for any level of violation, a larger sample size increases the power of the test. At any sample size, greater proportions of violators increase the power of the test. When approximately 20\% of a sample of 800 confessors falsely confess, lie on the treatment list, or change their responses to the non-sensitive items on the treatment list, the placebo test achieves 80\% power. The final panel shows that the power of the test is maximized when the variability in non-sensitive list items is minimized.
 
\begin{figure}[H]\centering
\caption{Power of Placebo Test I to Detect Violations of Assumptions 1-3}
\includegraphics[scale=.45]{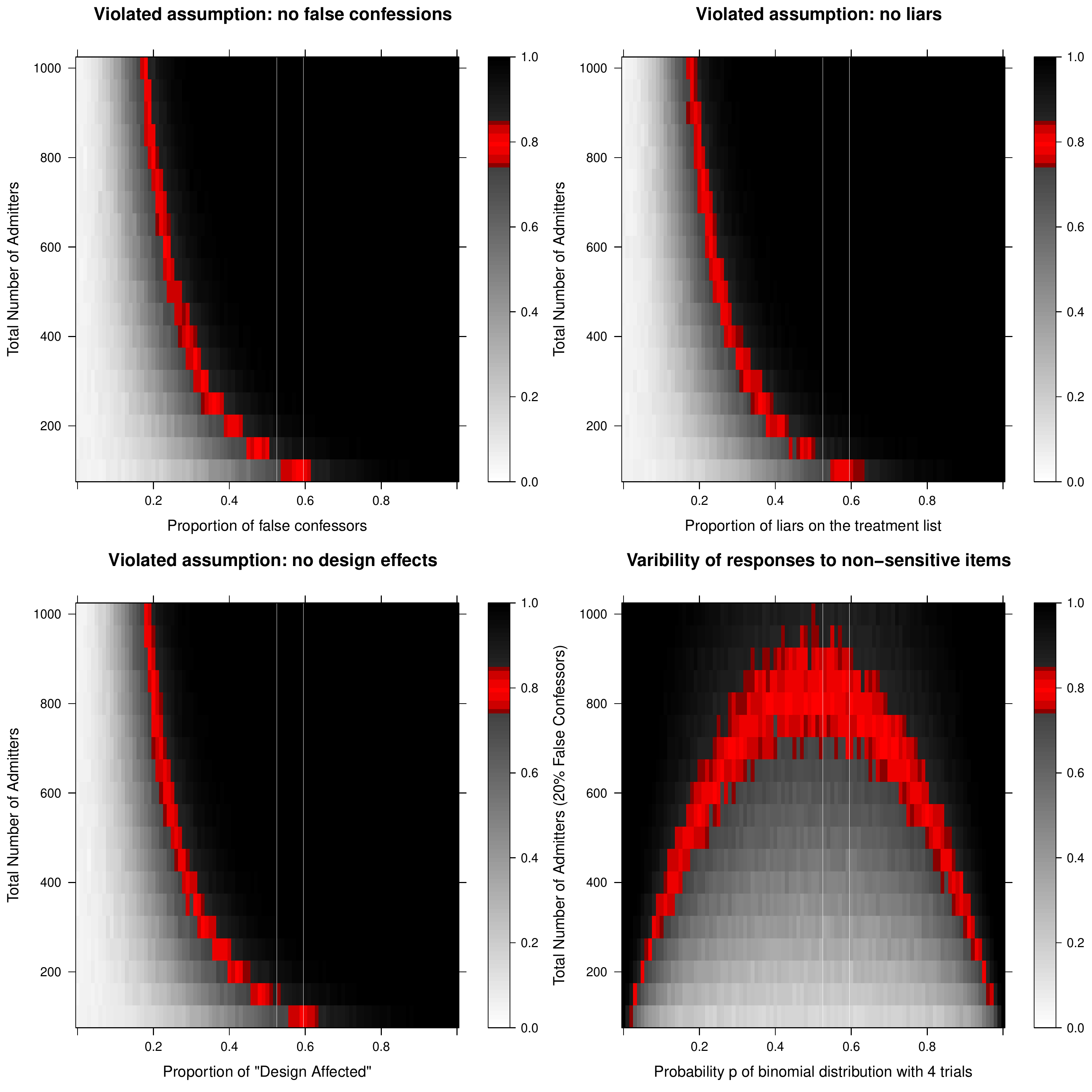}
\label{fig: powertiles}
\end{figure}

\section*{Appendix D: Placebo Test II Power Analysis}

The power of Placebo Test II to detect violations of Treatment Independence is equivalent to the power of test of a difference in proportions. Following the implementation in R (see accompanying text: Dalgaard 2008, p. 159), the power of the test for N = 100, 200, 500, and 1000 is shown in Figure~\ref{fig: powerIItiles} below. In all panels, the proportion of ``Yes'' responses in the treatment group is given on the x-axis, and the proportion in the control group is given on the y-axis. Darker coloring indicates higher power, with the band around the 0.80 power target shaded red.  

\begin{figure}[H]\centering
\caption{Power of Placebo Test II to Detect Violations of Treatment Independence}
\includegraphics[scale=.45]{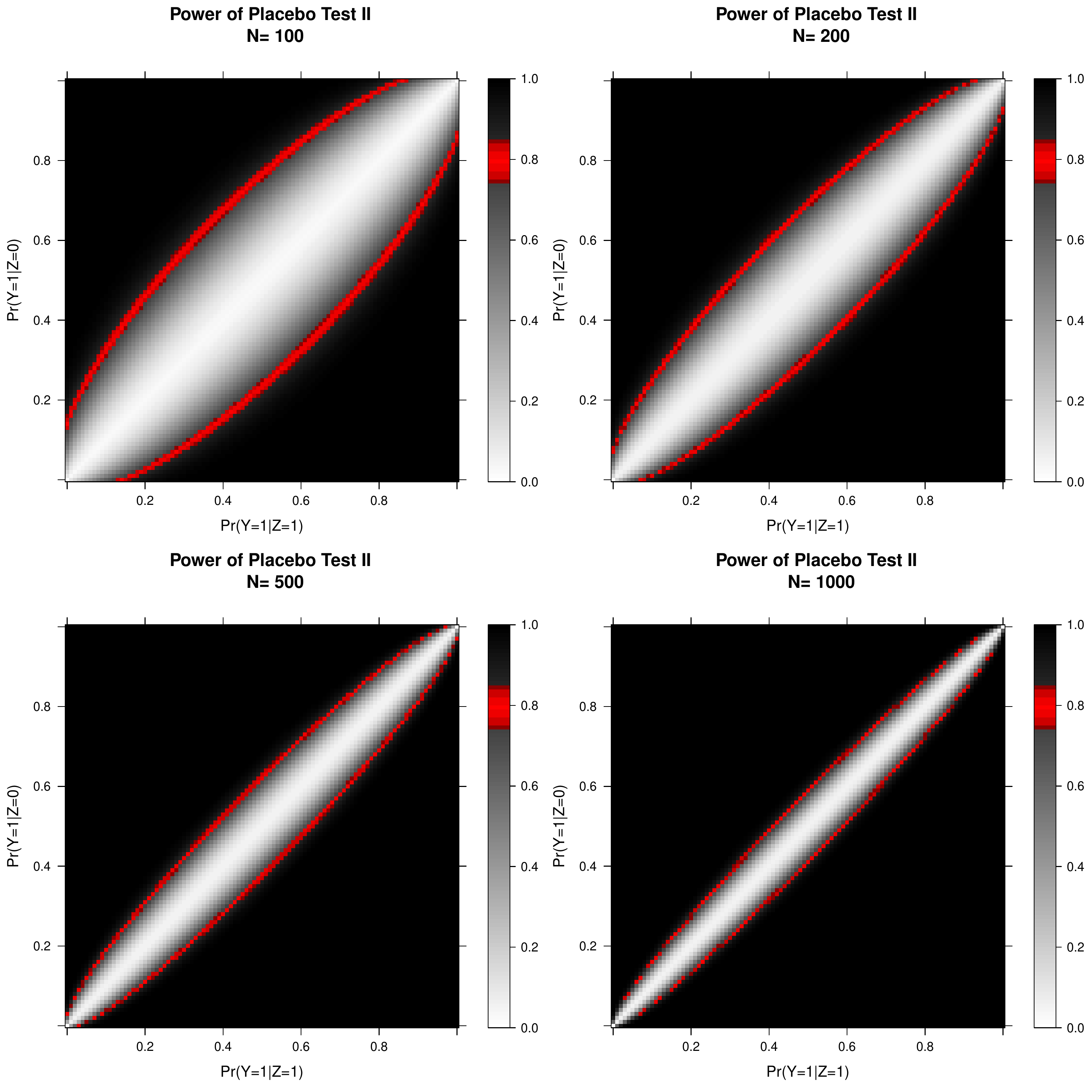}
\label{fig: powerIItiles}
\end{figure}

\section*{Appendix E: Factorial Design}

Subjects were randomly assigned to Study A (Direct Questions first) or Study B (List Experiments First).  Table~\ref{tab:factorialmain} below presents the differences in estimates obtained by the two studies.  The order in which the direct and list questions does not appear to have had a  significant impact on the Direct or Combined List estimates of prevalance, but does appear to have changed the Conventional List estimates in two cases.  The list experimental estimate of the proportion supporting Nuclear power is 24.9 percentage points higher when the direct questions are asked first, whereas the estimate of those watching CNN is 30.7 percentage points lower when the direct questions are asked first.  

\begin{table}[H]
\centering 
\caption{Factorial Design: Study A - Study B} 
\begin{tabular}{rrrrrrr}
&\multicolumn{2}{c}{Direct} & \multicolumn{2}{c}{Conventional List} & \multicolumn{2}{c}{Combined List}\\
\cmidrule(r){2-3} \cmidrule(r){4-5} \cmidrule(r){6-7}
 & Difference & SE & Difference & SE & Difference & SE \\ 
  \hline
Nuclear Power & 0.053 & 0.030 & 0.249 & 0.122 & 0.042 & 0.071 \\ 
  Public Transportation & -0.001 & 0.031 & -0.064 & 0.103 & 0.019 & 0.071 \\ 
  Spanish-speaking & -0.011 & 0.019 & -0.109 & 0.115 & -0.106 & 0.106 \\ 
  Muslim Teachers & 0.007 & 0.019 & 0.060 & 0.116 & 0.064 & 0.105 \\ 
  CNN & -0.052 & 0.031 & -0.307 & 0.146 & -0.054 & 0.095 \\ 
   \hline

   \hline
\end{tabular}
\label{tab:factorialmain}
\end{table}

\section*{Appendix F: Replication Study}

The replication study was conducted with a new pool of Mechanical Turk respondents ten months after the first study.  The experimental procedure was identical in every respect to the design described in Section~\ref{sec:application}. 1019 users started the survey, but 7 either did not complete the survey or failed the attention question, leaving 1012 complete cases. 506 subjects participated in Study A and 506 participated in Study B.  

The formats of the tables below follow those Section~\ref{sec:application}, facilitating comparisons. Of particular note are the results of the placebo tests, presented in Tables \ref{tab:studyAplaceboreplication}, \ref{tab:studyBplaceboreplication}, and \ref{tab:primingplaceboreplication}.  In the original study, two of five questions failed the placebo test in Study A and none of the question failed in Study B.  The replication shows a different pattern:  one of five fail in Study A, whereas four of five (at the 10\% level or greater) fail in Study B.  In the original study, both Nuclear Power and CNN failed the second placebo test, but in the replication, only CNN fails.  We interpret this result to mean that the effect of the treatment list on direct answers to the CNN is not a fluke due to sampling variability but rather a robust causal effect on the order of 10 percentage points.

\begin{table}[H]
\centering
\singlespacing
\caption{Study A (Direct First): Three Estimates of Prevalence}
\begin{tabular}{rp{1cm}p{1cm}p{1cm}p{1cm}p{1cm}p{1cm}c}
  \hline  
    \noalign{\vskip 1mm}  
   & \multicolumn{2}{c}{Direct} & \multicolumn{2}{c}{Standard List} & \multicolumn{2}{c}{Combined List} &  \multirow{2}{*}{\parbox{25mm}{\centering \% Reduction in Sampling Variance}}\\
         \noalign{\vskip 1mm} 
    \cmidrule(r){2-3} \cmidrule(r){4-5} \cmidrule(r){6-7}
  & \multicolumn{1}{c}{$\overline{Y}$} & \multicolumn{1}{c}{SE} & \multicolumn{1}{c}{$\hat\mu_S$} & \multicolumn{1}{c}{SE}& \multicolumn{1}{c}{$\hat\mu$} & \multicolumn{1}{c}{SE} &  \\
   \noalign{\vskip 1mm}  
  \hline
Nuclear Power & 0.646 & 0.021 & 0.575 & 0.086 & 0.647 & 0.049 & 67.326 \\ 
  Public Transportation & 0.555 & 0.022 & 0.635 & 0.073 & 0.653 & 0.048 & 56.643 \\ 
  Spanish-speaking & 0.061 & 0.011 & 0.118 & 0.078 & 0.104 & 0.073 & 11.798 \\ 
  Muslim Teachers & 0.083 & 0.012 & 0.034 & 0.078 & 0.036 & 0.072 & 14.651 \\ 
  CNN & 0.407 & 0.022 & 0.256 & 0.100 & 0.298 & 0.074 & 45.227 \\ 
   \hline
      \multicolumn{8}{l}{$n$ = 506 for all estimates}
\end{tabular}
\label{tab:studyAestsreplication} 
\end{table}

\begin{table}[H]
\centering
\singlespacing
\caption{Study A (Direct First): Placebo Test I} 
\begin{tabular}{rrrrr}
  \hline
      \noalign{\vskip 2pt}  
 & \multicolumn{1}{c}{$\hat \beta$} & \multicolumn{1}{c}{SE} &\multicolumn{1}{c}{$\Pr[\beta \neq 1]$} & \multicolumn{1}{c}{$n$} \\ 
  \hline
Nuclear Power & 0.938 & 0.104 & 0.549 & 327.000 \\ 
  Public Transportation & 0.949 & 0.089 & 0.566 & 281.000 \\ 
  Spanish-speaking & 0.816 & 0.311 & 0.554 & 31.000 \\ 
  Muslim Teachers & 1.048 & 0.297 & 0.872 & 42.000 \\ 
  CNN & 0.711 & 0.133 & 0.029 & 206.000 \\ 
   \hline
\end{tabular}
\label{tab:studyAplaceboreplication}
\end{table}

\begin{table}[H]
\centering
\caption{Study A (Directs First): Placebo Test II}
\begin{tabular}{rrr}
  \hline
 & Estimate & SE \\ 
  \hline
Nuclear Power & -0.059 & 0.043 \\ 
  Public Transportation & 0.017 & 0.044 \\ 
  Spanish-speaking & 0.031 & 0.021 \\ 
  Muslim Teachers & -0.009 & 0.025 \\ 
  CNN & 0.075 & 0.044 \\ 
   \hline

   \hline
\end{tabular}
\label{placebo2directsfirstreplication}
\end{table}

\begin{table}[H]
\centering
\singlespacing
\caption{Study B (Lists First): Three Estimates of Prevalence}
\begin{tabular}{rp{1cm}p{1cm}p{1cm}p{1cm}p{1cm}p{1cm}c}
  \hline  
    \noalign{\vskip 1mm}  
   & \multicolumn{2}{c}{Direct} & \multicolumn{2}{c}{Standard List} & \multicolumn{2}{c}{Combined List} &  \multirow{2}{*}{\parbox{25mm}{\centering \% Reduction in Sampling Variance}}\\
      \noalign{\vskip 1mm}  
    \cmidrule(r){2-3} \cmidrule(r){4-5} \cmidrule(r){6-7}
  & \multicolumn{1}{c}{$\overline{Y}$} & \multicolumn{1}{c}{SE} & \multicolumn{1}{c}{$\hat\mu_S$} & \multicolumn{1}{c}{SE}& \multicolumn{1}{c}{$\hat\mu$} & \multicolumn{1}{c}{SE} &  \\
   \noalign{\vskip 1mm}  
  \hline
Nuclear Power & 0.591 & 0.022 & 0.364 & 0.087 & 0.521 & 0.054 & 61.047 \\ 
  Public Transportation & 0.520 & 0.022 & 0.602 & 0.073 & 0.672 & 0.049 & 54.183 \\ 
  Spanish-speaking & 0.099 & 0.013 & 0.025 & 0.079 & 0.094 & 0.074 & 11.822 \\ 
  Muslim Teachers & 0.103 & 0.014 & 0.195 & 0.082 & 0.164 & 0.074 & 18.610 \\ 
  CNN & 0.457 & 0.022 & 0.586 & 0.107 & 0.604 & 0.075 & 50.158 \\ 
   \hline
      \multicolumn{8}{l}{$n$ = 506 for all estimates}
\end{tabular}
\label{tab:studyBestsreplication}
\end{table}

\begin{table}[H]
\centering
\singlespacing
\caption{Study B (Lists First): Placebo Test I}
\begin{tabular}{rrrrr}
  \hline
      \noalign{\vskip 2pt}  
 & \multicolumn{1}{c}{$\hat \beta$} & \multicolumn{1}{c}{SE} &\multicolumn{1}{c}{$\Pr[\beta \neq 1]$} & \multicolumn{1}{c}{$n$} \\ 
  \hline
Nuclear Power & 0.739 & 0.110 & 0.018 & 299.000 \\ 
  Public Transportation & 0.795 & 0.097 & 0.034 & 263.000 \\ 
  Spanish-speaking & 0.301 & 0.262 & 0.008 & 50.000 \\ 
  Muslim Teachers & 1.005 & 0.297 & 0.987 & 52.000 \\ 
  CNN & 0.759 & 0.141 & 0.087 & 231.000 \\ 
   \hline
\end{tabular}
\label{tab:studyBplaceboreplication}
\end{table}

\begin{table}[H]
\centering
\caption{Study B (Lists First): Placebo Test II}
\begin{tabular}{rrr}
  \hline
 & Estimate &  SE \\ 
  \hline
Nuclear Power & -0.003 & 0.044 \\ 
  Public Transportation & 0.075 & 0.045 \\ 
  Spanish-speaking & -0.001 & 0.027 \\ 
  Muslim Teachers & 0.039 & 0.027 \\ 
  CNN & 0.092 & 0.044 \\ 
   \hline

   \hline
\end{tabular}
\label{tab:primingplaceboreplication}
\end{table}

\begin{table}[H]
\centering 
\caption{Factorial Design: Study A - Study B} 
\begin{tabular}{rrrrrrr}
&\multicolumn{2}{c}{Direct} & \multicolumn{2}{c}{Conventional List} & \multicolumn{2}{c}{Combined List}\\
\cmidrule(r){2-3} \cmidrule(r){4-5} \cmidrule(r){6-7}
 & Difference & SE & Difference & SE & Difference & SE \\ 
  \hline
Nuclear Power & 0.055 & 0.031 & 0.210 & 0.123 & 0.127 & 0.073 \\ 
  Public Transportation & 0.036 & 0.031 & 0.032 & 0.103 & -0.019 & 0.069 \\ 
  Spanish-speaking & -0.038 & 0.017 & 0.093 & 0.111 & 0.010 & 0.104 \\ 
  Muslim Teachers & -0.020 & 0.018 & -0.161 & 0.113 & -0.128 & 0.104 \\ 
  CNN & -0.049 & 0.031 & -0.330 & 0.146 & -0.306 & 0.105 \\ 
   \hline

   \hline
\end{tabular}
\label{tab:factorialreplication}
\end{table}

\end{document}